\documentclass[a4paper,onecolumn,11pt,accepted=2021-09-24]{quantumarticle}
\pdfoutput=1
\usepackage[numbers,sort&compress]{natbib}
\usepackage[utf8]{inputenc}
\usepackage[english]{babel}
\usepackage[T1]{fontenc}
\usepackage{amsmath}
\usepackage{hyperref}
\usepackage{tikz}
\usepackage{lipsum}
\usepackage{float}
\usepackage{graphicx}
\usepackage{mathtools}
\usepackage{nicematrix}
\usepackage{arydshln}
\usepackage{amsfonts}
\usepackage{braket}
\usepackage{dsfont}
\usepackage[english]{babel}
\usepackage{amssymb, amsmath}
\usepackage{amsthm}
\usepackage{changes}
\usepackage{hyperref}
\usepackage{cleveref}
\Crefname{equation}{Eq.}{Eqs.}
\newtheorem{definition}{Definition}[section]
\newtheorem{theorem}{Theorem}[section]
\newtheorem{lemma}{Lemma}[section]
\newtheorem{corollary}{Corollary}[section]

\newcommand{\be}{\begin{equation}} 
\newcommand{\bes}{\begin{equation*}}
\newcommand{\ee}{\end{equation}}
\newcommand{\ees}{\end{equation*}}
\newcommand{\bematrix}{\left(\begin{matrix}}
\newcommand{\ematrix}{\end{matrix}\right)}

\newcommand{\rvline}{\hspace*{-\arraycolsep}\vline\hspace*{-\arraycolsep}}
\newcommand*{\matminus}{%
  \leavevmode
  \hphantom{0}%
  \llap{%
    \settowidth{\dimen0 }{$0$}%
    \resizebox{1\dimen0 }{\height}{$-$}%
  }%
}

\begin{document}

\title{Entangled symmetric states and copositive matrices}
\date{December 11, 2020}
\author{Carlo Marconi}
\affiliation{F\'isica Te\`{o}rica: Informaci\'o i Fen\`{o}mens Qu\`{a}ntics.  Departament de F\'isica, Universitat Aut\`{o}noma de Barcelona, 08193 Bellaterra, Spain}
\author{Albert Aloy}
\affiliation{ICFO - Institut de Ci\`{e}ncies Fot\`{o}niques, The Barcelona Institute of Science and Technology, 08860 Castelldefels (Barcelona), Spain}
\author{Jordi Tura}
\affiliation{Max-Planck-Institut f\"ur Quantenoptik, Hans-Kopfermann-Str. 1, 85748 Garching, Germany}
\affiliation{Instituut-Lorentz, Universiteit Leiden, P.O. Box 9506, 2300 RA Leiden, The Netherlands}
\author{Anna Sanpera}
\affiliation{F\'isica Te\`{o}rica: Informaci\'o i Fen\`{o}mens Qu\`{a}ntics.  Departament de F\'isica, Universitat Aut\`{o}noma de Barcelona, 08193 Bellaterra, Spain}
\affiliation{ICREA, Pg. Llu\'{\i}s Companys 23, 08010 Barcelona, Spain}

\begin{abstract}
	\noindent   
Entanglement in symmetric quantum states and the theory of copositive matrices are intimately related concepts. For the simplest symmetric states, i.e., the diagonal symmetric (DS) states, it has been shown that there exists a correspondence between exceptional (non-exceptional) copositive matrices and non-decomposable (decomposable) entanglement witnesses (EWs). Here we show that EWs of symmetric, but not DS, states can also be constructed from extended copositive matrices, providing new examples of bound entangled symmetric states, together with their corresponding EWs, in arbitrary odd dimension.  

\end{abstract}

\maketitle
	
\noindent 
Entanglement and symmetry lie at the heart of quantum theory. Symmetries reflect fundamental laws of Nature and are intrinsically present in systems of physical interest. Moreover, states possessing some symmetry, typically admit a simplified mathematical description as compared to the one of generic states, a fact that usually translates into a more feasible way of characterizing their properties.

Quantum correlations are an intrinsic property of composite systems. Entanglement, in particular, is regarded as the most significant feature of quantum physics, not only because it provides unique insights into the fundamental principles of our physical world, but also because it represents a resource that allows to perform several tasks that would be, otherwise, impossible.

Since the birth of quantum information theory, huge efforts have been devoted to characterize and quantify entanglement (see e.g. \cite{horodecki2009quantum}). Along the years, it has become clear that entanglement characterization is a challenging task. Moreover, it cannot be quantified by a unique measure. The exception lies in (bipartite) pure entangled states where it is trivial to determine if the state is entangled and entanglement entropy is the unique measure needed. Interestingly, in the asymptotic limit, for a sufficient number of copies of the system, the entanglement entropy measures the resource interconversion capacity between different pure states, within the paradigm of local operations and classical communication \cite{bennett1996concentrating}. However, already in the case of bipartite mixed states, two of such measures are needed in order to quantify this interconversion rate: the entanglement of formation and the entanglement of distillation.

A closely related, although inherently different, approach is the characterization of entangled states independently of any measure or their usefulness for a specific task. This problem has been shown to be, in general, NP-hard \cite{gurvits2003classical}. However, partial characterization has been achieved by means of criteria that provide necessary, but not sufficient, conditions to determine if a given state is entangled or not. The most powerful of such criteria, formulated in terms of linear positive maps, is the positivity under partial transposition (PPT) \cite{peres1996separability}, which is the paradigmatic example of a positive, but not completely positive, map. States that do not fulfill the PPT criterion are entangled but the converse is not true, except for few cases. In this regard, quantum maps and their associated entanglement witnesses (EWs), provide the strongest criteria for entanglement characterization: a quantum state is entangled if, and only if, there exists an EW that detects it \cite{terhal2000entanglement,lewenstein2000optimization, chruscinski2014entanglement}. Crucially, in order to characterize entanglement in states that do not break the PPT criterion (PPT entangled states or PPTES), it is necessary to construct non-decomposable EWs \cite{lewenstein2001characterization}. Interestingly, EWs have been shown to provide also a measure of entanglement which is upper and lower bounded by other entanglement measures \cite{brandao2005quantifying}. 

Nowadays, it is still unclear whether, in general, the problem of entanglement characterization remains equally hard for systems displaying some symmetries \cite{vollbrecht2001entanglement,toth2010separability,terhal2000entanglement,eggeling2001separability}. A possible approach in this direction is to investigate if, and how, symmetries can help to construct EWs for such systems.

In this work, we focus on the entanglement characterization of permutationally invariant systems and, more specifically, on the class of the so-called symmetric states. These provide a natural description for sets of indistinguishable particles, i.e., bosons. To this aim, we derive a method to construct specific EWs for such states using the theory of copositive matrices.  
As shown previously by some of us \cite{tura2018separability}, the characterization of entanglement for some particularly simple symmetric states that are mixtures of projectors on symmetric states, also called DS states, can be recast as the problem of checking membership to the cone of the so-called completely-positive matrices or to its dual, the cone of copositive matrices. The equivalence between these two problems allows to establish a correspondence between entanglement characterization for DS states and non-convex quadratic optimization \cite{tura2018separability}. In particular, it provides a very efficient method to detect states that are PPT entangled. Here, we further extend this mapping by constructing EWs that detect symmetric PPTES that are not DS. First, we derive under which conditions a copositive matrix leads to an EW for symmetric states. Then, we show explicitly how to derive, from such copositive matrices, both decomposable and non-decomposable EWs. Finally, we use these EWs to generate several families of symmetric PPTES. It is important to remark that, to the best of our knowledge, there are not known examples of non-decomposable EWs for generic symmetric states. The very few examples of symmetric PPTES in the literature\cite{toth2010separability}, have been found numerically using weaker entanglement criteria like, e.g., the range criterion applied to edge states\cite{lewenstein2001characterization,sanpera2001schmidt, clarisse20065, kye2012classification, chen2011description, magne2010numerical}. Our work thus, offers a complementary approach to the study of entanglement characterization in symmetric states.
The paper is organized as follows: in \Cref{sec:Preliminaries}, we introduce basic concepts concerning the definition and properties of symmetric states in $\mathbb C^{d}\otimes \mathbb C^{d}$, and the mapping between copositive matrices and EWs. 
In this section, we review as well some of our previous results for DS states \cite{tura2018separability} that represent the starting point of our work. In \Cref{sec:EWcopositive}, we show how to derive decomposable and non-decomposable EWs for symmetric states from copositive matrices and provide some explicit constructive examples. 
In \Cref{sec:symBoundEntanglement}, we focus on symmetric, but not DS, PPTES, showing the existence of PPTES in arbitrary dimensions and introducing examples of such states in different dimensions ($d=3,4,5,7$). Finally, in \Cref{sec:Conclusions} we summarize our findings, present some conjectures derived from our analysis and list some open questions for further research.

\section{Basic concepts}	
\label{sec:Preliminaries}

 We start by introducing the notation used throughout the manuscript along with some basic concepts and definitions regarding symmetric states, entanglement witnesses and copositive matrices.

\subsection{Symmetric states}
Henceforth, we focus on bipartite systems. 
Let $\mathcal H=\mathbb{C}^{d}\otimes \mathbb{C}^{d}$, be the finite dimensional Hilbert space of two qudits, and $\mathcal{B}(\mathcal{H})$ the set of its bounded operators. The symmetric subspace $\mathcal {S} \subset\mathcal{H}$, is the convex set formed by the (normalized) pure states $\ket{\Psi_S}\in \mathcal H$ that are invariant under the exchange of parties. Symmetric states can be mapped to spin systems and, moreover, they span the subspace of maximum spin in the Schur-Weyl representation. The so-called Dicke states form a convenient basis of the symmetric subpace, i.e.,  $\mathcal{S}\equiv \text{span} \{\ket{D_{ij}}\}$, where $\ket{D_{ij}} =\ket{D_{ji}}\equiv(\ket{ij} + \ket{ji})/\sqrt{2}$ for $i \neq j,\; \ket{D _{ii}} \equiv \ket{ii}$ and $\{\ket{i}\}_{i=0}^{d-1}$ is an orthonormal basis of $\mathbb{C}^{d}$. 
Note the reduced dimensionality of the symmetric subspace, i.e., $\dim(\mathcal{S})=d(d+1)/2 <d^2$. In an abuse of language we denote symmetric quantum states, $\rho_S\in \mathcal{B}(\mathcal{S})$, as the convex hull of projectors onto pure symmetric normalized states, i.e.,
$\rho_{S} =\sum_{k} p^{(k)}_{S} \ket{\Psi^{(k)}_S}\bra{\Psi^{(k)}_S}$, with $p^{(k)}_{S} \geq 0$, $\sum_{k} p^{(k)}_{S} =1$ and 
$\ket{\Psi^{(k)}_{S}}=\sum_{ij} c^{(k)}_{ij} \ket{D_{ij}}$, $c^{(k)}_{ij}\in \mathbb{C} $.\\ Thus, any $\rho_S\in \mathcal{B}(\mathcal{S})$ is a positive semidefinite operator ($\rho_S\succeq 0$) with unit trace ($\mathrm{Tr} (\rho_S) = 1$), fulfilling $\Pi_{S} \rho_{S} \Pi_{S} = \rho_{S} \Pi_{S} = \Pi_{S} \rho_{S}= \rho_{S}$, where $\Pi_{S} = \frac{1}{2}(\mathds{1} + F)$ is the projector onto the symmetric subspace and $F = \sum_{i,j = 0}^{d-1} \ket{ij}\!\bra{ji}$ is the flip operator. Using the Dicke basis, symmetric quantum states can be compactly 
expressed as follows: 
 {\definition{
 \label{def:SS}
Any bipartite symmetric state,  $\rho_{S}\in \mathcal{B}(\mathcal{S})$, can be written as 
\begin{equation}
\rho_{S} = \sum_{\substack{0\leq i\leq j<d\\
                  0\leq k\leq l<d}}
               \left(\rho_{ij}^{kl} \ket{D_{ij}}\bra{D_{kl}} + \emph{h.c.} \right)~,
\end{equation}
with $ \rho_{ij}^{kl} \in \mathbb{C}$. Notice that, due to the symmetry of the Dicke states, it holds that $ \rho_{ij}^{kl} = \rho_{ji}^{kl}= \rho_{ij}^{lk}= \rho_{ji}^{lk} ~\forall i,j,k,l$.\\
}}

\noindent Convex mixtures of projectors onto Dicke states are denoted as diagonal symmetric (DS) states, since they are diagonal in the Dicke basis. They form a convex subset of $\mathcal{S}$ and are particularly relevant for our analysis.
{ \definition{
 \label{def:DSS}
Any DS state, $\rho_{\small{DS}} \in \mathcal{B}(\mathcal{S})$, is of  the form
 \begin{equation}
  \rho_{DS} = \sum_{0 \leq i \leq j < d} p_{ij}\ket{D_{ij}}\bra{D_{ij}},
  \label{eq:DSS}
 \end{equation}
with $p_{ij} \geq 0,\; \forall\; i,j$ and $\sum_{ij} p_{ij}=1$.}}

 \begin{lemma}
 Every symmetric state, $\rho_{S} \in \mathcal{B}(\mathcal{S})$, can be written as the sum of a DS state, $\rho_{DS}$, and a traceless symmetric contribution, $\sigma_{CS}$, which contains all coherences between Dicke states, i.e.,
\begin{equation}
\label{sym}
\rho_{S} = \rho_{DS}+\sigma_{CS}= 
\sum_{0\le i \le j < d} p_{ij} \ket{D_{ij}} \bra{D_{ij}} + \underset{{(i,j)\neq(k,l)}}{\sum_{ij}\sum_{kl}} \left( {\alpha_{ij}^{kl}} \ket{D_{ij}} \bra{D_{kl}} + \emph{h.c.} \right),
\end{equation}
\noindent with  ${\alpha_{ij}^{kl}} \in \mathbb{C}$ and ${\alpha_{ij}^{kl}} = ({\alpha_{kl}^{ij}})^*$.
\end{lemma}

\subsection{ Separability, EWs and copositivity}

{\definition{A bipartite symmetric state $\rho_{S} \in \mathcal{B}(\mathcal{S})$ is separable (not entangled) if it can be written as a convex combination of projectors onto pure symmetric product states, \textit{i.e.,}}
\begin{equation}
\rho_{S}= \sum_{i} p_i \ket{e_{i}e_{i}}\bra{e_{i}e_{i}},
\label{eq:sep}
\end{equation}
with $p_i\geq 0$, $\sum_i p_i=1$ and $\ket{e_{i}}=\sum_{i} e_{i}^{(k)}\ket{k}$, where $e_{i}^{(k)} \in \mathbb{C}$ and $\{\ket{k}\}_{k=0}^{d-1}$ is an orthonormal basis in $\mathbb{C}^d$. If a decomposition of this form does not exist, then $\rho_{S}$ is \textit{entangled}.}

\begin{figure}[H]
	\centering
	\includegraphics[width=4cm]{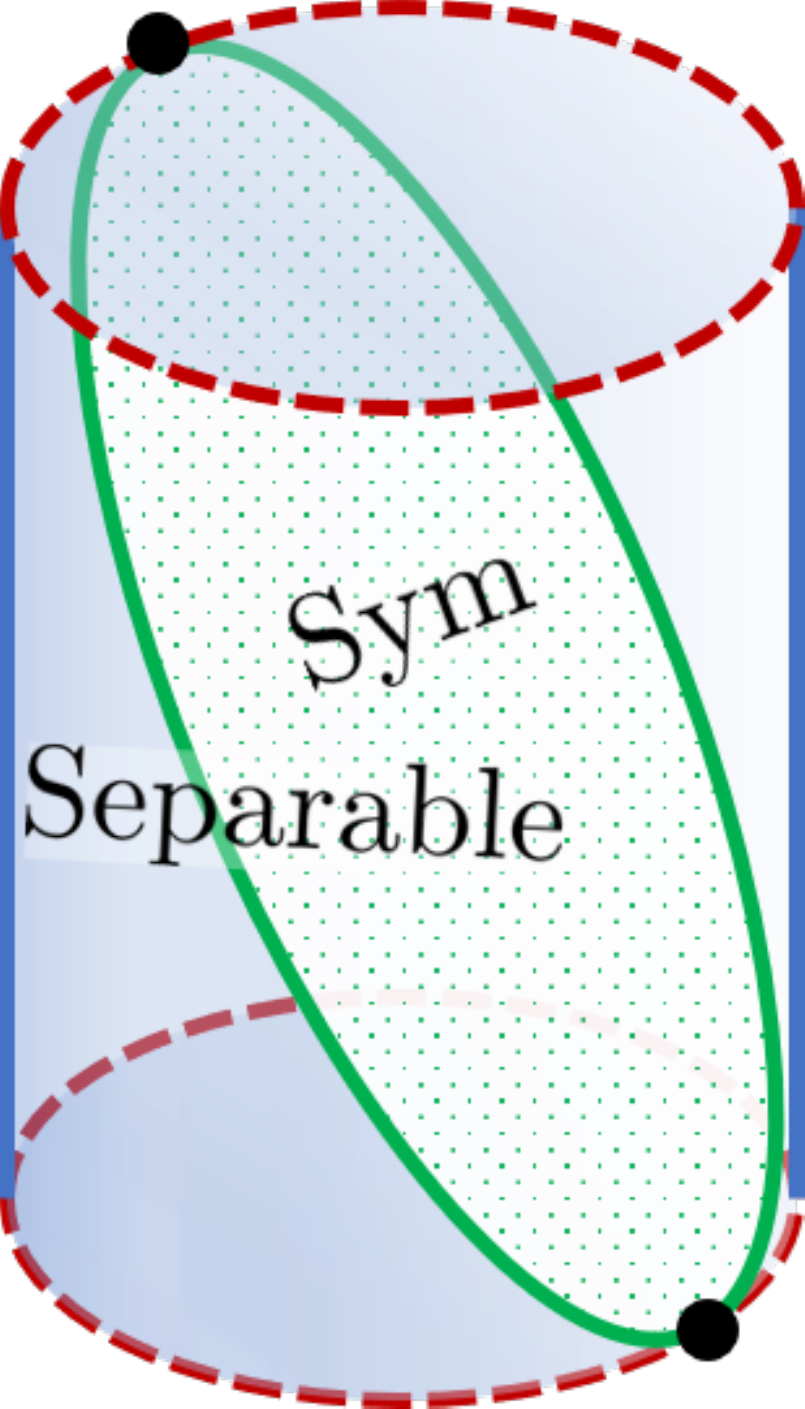}
	\caption{Pictorial representation of the set of bipartite symmetric separable states $\mathcal{D}_S$ (Sym) embedded into the set of bipartite separable states. The cylinder represents the separable set $\mathcal{D}$. The discontinuous (red) line corresponds to the extremal points (of the form $\ket{e,f}$) generating the set and the continuous (blue and green) lines corresponds to the respective boundaries (necessarily requiring description as density matrices with rank $>1$ but not maximal). Both the separable and the symmetric separable sets share extremal points of the form $\ket{e,e}$, here represented by the black dots.}
	\label{fig:separable}
\end{figure}

We denote by ${\mathcal D}$, the compact set of separable quantum states and by $\mathcal D_{S}$, its analogous symmetric counterpart, which is also compact (see Fig.\ref{fig:separable}). As a consequence of the Hahn-Banach theorem, the set $\mathcal D_{S}$ admits also a dual description in terms of its dual cone, $\mathcal{P_{S}}$, defined as the set of the operators $W$ fulfilling
\be
\mathcal {P}_{S}= \{ W=W^{\dagger}\;\;\mbox{s.t}\; \langle W, \rho \rangle  \geq 0 ~, \forall \rho_{S} \in \mathcal{D}_{S}\}~,
\ee
where $\langle W, \rho \rangle \equiv \mbox{Tr} (W^{\dagger}\rho)$ is the Hilbert-Schmidt scalar product.

{\definition{A Hermitian operator, $W \in \mathcal P_{S}$, is an entanglement witness (EW) of symmetric states if, and only if, it satisfies the following properties: 
\begin{itemize}
	\item[1.]  $\emph{Tr}(W\rho_{S})\geq 0, ~ \forall   \;\rho_{S}\in\mathcal D_{S}$ ,~
	\item[2.] There exists at least one symmetric state $\rho_{S}$ such that $\emph{Tr}(W \rho_{S}) < 0$~. 
\end{itemize}}}

\noindent Notice that, by definition, the set of separable symmetric states, $\mathcal D_{S}$, satisfies the inclusion $\mathcal D_{S}\subset \mathcal D$, but $\mathcal P \subset \mathcal P_S$, where $\mathcal {P}$ is the dual cone of the convex set $\mathcal D$, i.e., 
\begin{equation}
    \mathcal {P}= \{ W=W^{\dagger}\;\; \mbox{s.t}. \; \langle W, \rho \rangle \geq 0~, \forall \rho \in \mathcal{D}\}~.
\end{equation}

In other words, any EW acting on $\mathcal{H}$ that detects an entangled state belongs to $\mathcal P_S$, but the converse is not necessarily true (see Fig.\ref{fig:planetary}).
\begin{figure}[H]
	\centering
	\includegraphics[width=7cm]{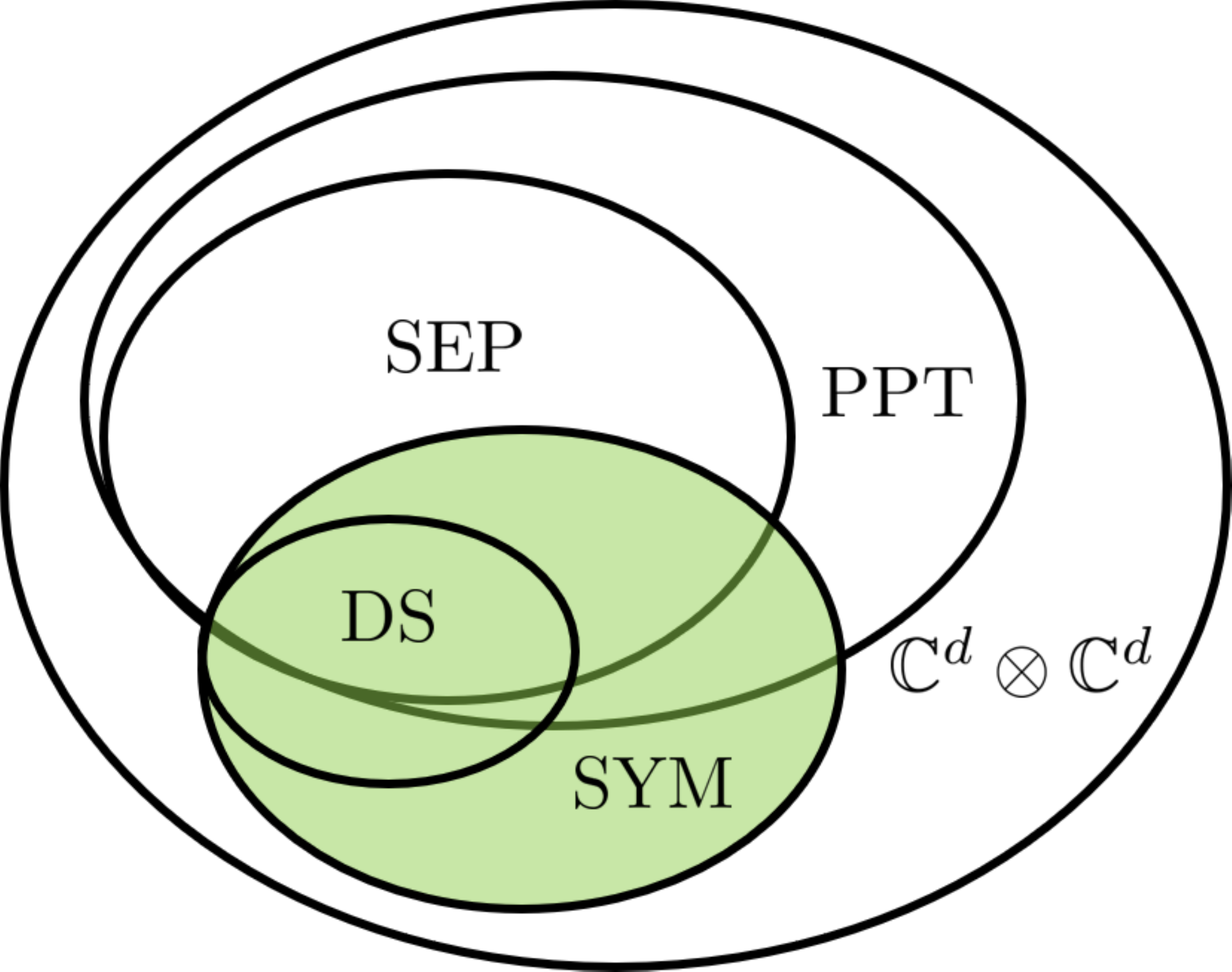}
	\caption{Pictorial structure of the quantum states in $\mathbb{C}^d\otimes\mathbb{C}^d$ for $d>5$. Each set contains the sets displayed inside. The colored region (green) represents the set of symmetric states (SYM). Note that, while for $d>5$ there exist diagonal symmetric (DS) states that are PPT-entangled, as represented in the figure, for $d<5$ all PPT-entangled DS state are necessarily separable (SEP) (see the text for details).}
	\label{fig:planetary}
\end{figure}
\noindent EWs are either decomposable or non-decomposable.
\begin{definition}
\label{dec}
An EW, $W$, is said to be decomposable (non-decomposable) if it can (cannot) be written as
\begin{equation}
W = P + Q^{T_{B}}, 
\end{equation}
with $P,Q \succeq 0$. Here $T_{B} \equiv \mathds{1}_{A} \otimes T$ denotes the partial transposition w.r.t. subsystem $B$, where $T$ stands for the usual matrix tranposition.\\
\end{definition}
\noindent It is easy to show that non-decomposable EWs are the only candidates to detect PPT entanglement. In fact, given a PPTES $\rho$, for any decomposable EW $W$, it is
\begin{equation}
\mbox{Tr}(W \rho) = \mbox{Tr}(P \rho) + \mbox{Tr}(Q^{T_{B}}\rho) = \mbox{Tr}(P \rho) + \mbox{Tr}(Q \rho^{T_{B}}) \geq 0~,
\end{equation}
\noindent where we have used the properties of the trace and the positive semidefiniteness of the operators $P$ and $Q$.\\
\noindent In particular from Def.(\ref{dec}) it follows that a EW is non-decomposable iff it detects at least one PPTES.\\

Remarkably, despite the apparent simplicity of the symmetric subspace due to its reduced dimensionality ($d(d+1)/2$ instead of $d^{2}$), entanglement characterization remains, in general, an open problem. For generic symmetric states, sparsity is preserved when the state is expressed in the computational basis but it is lost when the partial transposition is performed. However, for DS states, the corresponding partial transpose remains highly sparse and can be reduced to an associated matrix, $M_d(\rho_{DS})$, of dimensions $d \times d$, while generically $\rho_{S}^{T_{B}}$ is a matrix of dimension $d^2 \times d^2$.
{\definition {The partial transpose of every $\rho_{DS} \in \mathcal{B}(\mathcal{S})$ has the form
\begin{equation}
\label{transco}
\rho_{DS}^{T_{B}} = M_{d}(\rho_{DS}) \bigoplus_{0\le i \neq j < d} \left(p_{ij}/2\right)~,
\end{equation}
where  $M_{d}(\rho_{DS})$, which arises from the partially transposed matrix of a DS state in the computational basis, is defined as the $d\times d$ matrix with (non-negative) entries
\begin{equation}
\label{mr}
M_{d}(\rho_{DS}) :=
\bematrix
p_{00} 	& p_{0 1}/2  & \cdots & p_{0, d-1}/2 \\
p_{0 1}/2  & p_{11} & \cdots & p_{1, d-1}/2 \\
\vdots 	& \vdots & \ddots & \vdots \\
p_{0, d-1}/2 & p_{1, d-1}/2 & \cdots & p_{d-1, d-1} 
\ematrix ~.
\end{equation}}}

As shown in previous works \cite{yu2016separability,tura2018separability}, 
deciding if a DS state $\rho_{DS}$ is separable, is equivalent to check the membership of $M_{d}( \rho_{DS})$ to the cone of completely positive matrices $\mathcal {CP}_{d}$, i.e., the cone formed by those $d \times d$ matrices $A_{d}$ that admit a decomposition  of the type $A _{d}= B B^{T}$, where $B$ is a $d \times k$ matrix, for some $k>1$, with $B_{ij}\geq 0, B_{ij} \in \mathbb{R}$. Thus, if $\rho_{DS}$ is separable, then its associated matrix of Eq.(\ref{mr}) must satisfy $M_{d}(\rho_{DS})=B B^{T}$ \cite{yu2016separability}. This correspondence can be recast, equivalently, in the dual cone of $\mathcal {CP}_{d}$, i.e., in the cone  $\mathcal {COP}_{d}$ of copositive matrices.  As a result, copositive matrices act effectively as EWs for DS states. 
Below we provide the definition of a copositive matrix together with some properties that will be useful in the following.  

{\definition{ A real symmetric matrix, $H$, is copositive if, and only if, $\vec{x}^{T} H \vec{x} \geq 0,~ \forall \vec{x} \geq 0$ component-wise}}.\\

It is easy to see that the diagonal elements of a copositive matrix must be non negative, i.e., $H_{ii} \ge 0$, while negative elements $H_{ij}$ must fulfill $\sqrt{H_{ii} H_{jj}} \ge -H_{ij}$. Clearly, every positive semidefinite matrix is also copositive but the converse is, generically, not true. In fact, testing membership to the cone of copositive matrices is known to be a co-NP-hard problem \cite{murty1985some}\footnote{The co-NP problems are the complementary of the decision problems in NP.}, and only for $d\leq 4$, copositivity can be assessed analytically \cite{ping1993criteria,hiriart2010variational}. 

Finally, among copositive matrices, we distinguish extreme and exceptional copositive matrices that stand out for their impossibility to be decomposed.

{\definition{ A $d\times d$ copositive matrix $H$ is said to be extreme if $H = H_{1} + H_{2}$ with $H_{1}, H_{2}$ copositive, implies $H_{1} = a H, H_{2} = (1-a)H$ for all $a \in [0,1]$}}. 

{\definition{ A $d\times d$ copositive matrix $H$ is said to be \textit{exceptional} if, and only if, $H$ cannot be decomposed as the sum of a positive semidefinite matrix ($\mathcal{PSD}_{d}$), and a symmetric entry-wise non-negative matrix ($\mathcal{N}_{d}$),  i.e., $H \in \mathcal {COP}_{d} \setminus (\mathcal{PSD}_{d}+ \mathcal{N}_{d})$}}\\

\noindent Remarkably, it has been shown that for $d < 5$ there are no exceptional copositive matrices, 
meaning that, in this case, $ \mathcal {COP}_{d} = \mathcal{PSD}_{d}+ \mathcal{N}_{d}$ \cite{diananda1962non}. In
Fig.\ref{fig:cop} we illustrate, schematically, the relation between the aforementioned classes of copositive matrices.

\begin{figure}[H]
	\centering
	\includegraphics[width=6.5cm]{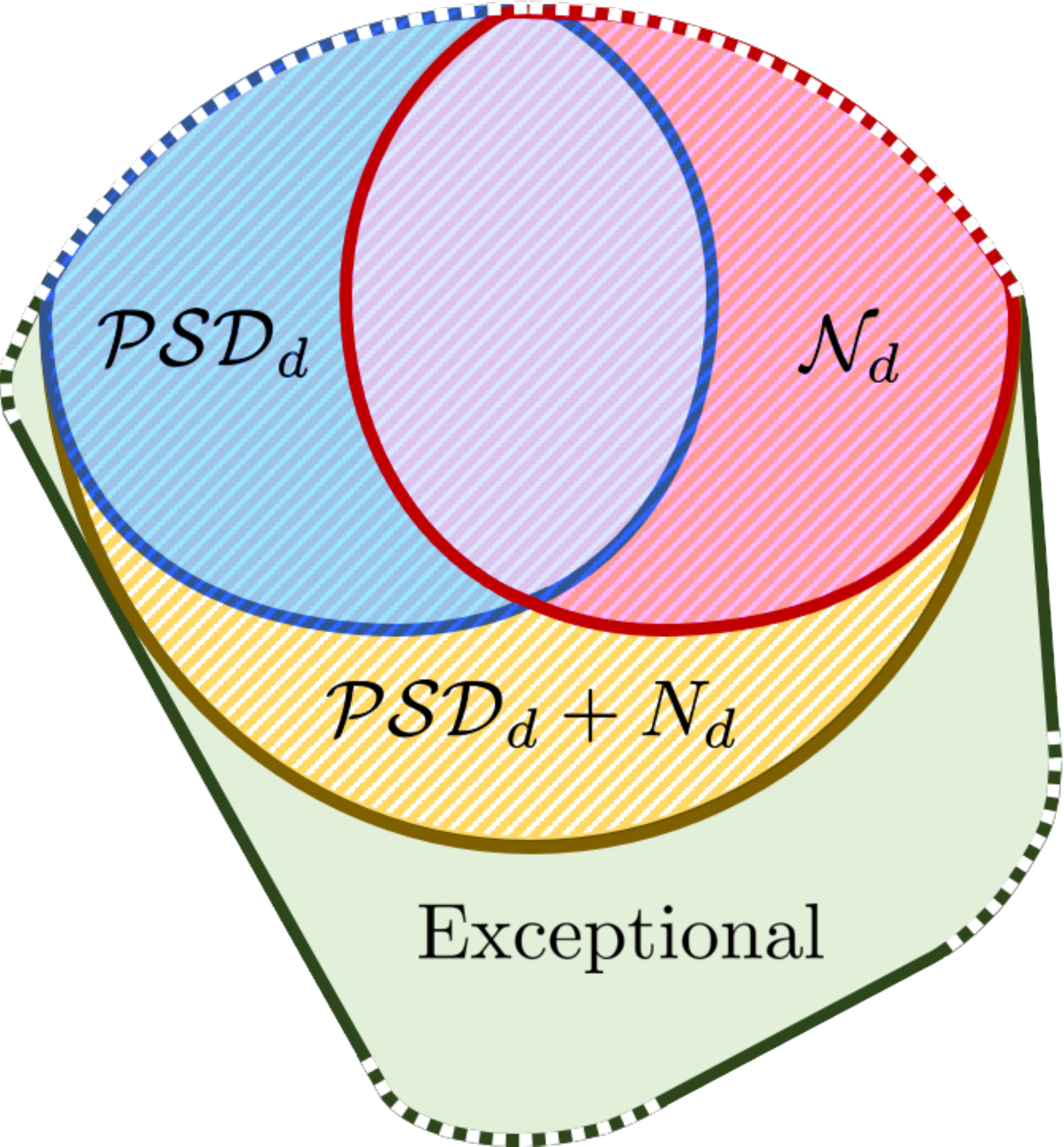}
	\caption{Pictorial representation of the cone of copositive $\mathcal{COP}_{d}$ and the cones $\mathcal{PSD}_{d}$ and $\mathcal{N}_{d}$. The striped region has been overmagnified for clarity and represents the convex hull of the cones $\mathcal{PSD}_{d}$ and $\mathcal{N}_{d}$, denoted as $\mathcal{PSD}_{d}+\mathcal{N}_{d}$. Note that exceptional copositive matrices exist only for $d>5$ (green). Extremal copositive matrices lie at the border of the cone $\mathcal{COP}_{d}$.}
	\label{fig:cop}
\end{figure}
 
\section {Copositive matrices as EWs}
\label{sec:EWcopositive}

Using the above definitions we can now show which copositive matrices lead to EWs. We prove explicitly how to construct a decomposable (non-decomposable) EW from a non-exceptional (exceptional) copositive matrix. Since decomposable EWs cannot detect bound entanglement, one is tempted to believe that separability in the symmetric subspace is equivalent to the analysis of exceptional copositive matrices. However, as we shall see later, this is not necessarily the case, and non-exceptional copositive matrices also play a relevant role in detecting bound entanglement.   
Our findings are summarized in the following theorems.

\begin{theorem}
\label{thext}
Each copositive matrix $H =\sum_{i,j=0}^{d-1}  H_{ij} \ket{i}\!\bra{j}$, 
with at least one negative entry $H_{mn}=H_{nm}<0$ ($m\neq n$), leads to an EW on $\mathcal{S}$ of the form $W=(H^{ext})^{T_{B}} =\sum_{{i,j=0}}^{d-1} H_{ij}\ket{ij}\!\bra{ji}$.
\end{theorem} 

\begin{proof} 

(i) We extend $H$ to the symmetric subspace as $H^{ext} = \sum_{i,j= 0}^{d-1} H_{ij} \ket{i}\bra{j} \otimes \ket{i}\bra{j}$, and denote $W=(H^{ext} )^{T_B}$. It is straightforward to show that $\mbox{Tr}(W \rho_{S}) \ge 0$ for all $ \rho_{S} \in \mathcal{D}_{S}$ since  $\bra{e e} W \ket{e e} = \braket{e e^{*}| H^{ext}| e e^{*}}=\sum_{ij}  |c_{i}|^{2} H_{ij} |c_{j}|^{2}= \vec{x}^{T} H \vec{x} \geq 0$, for every state $\ket{e} = \sum_{i=0}^{d-1} c_{i} \ket{i}$, where $c_{i}\in \mathbb C$, and $\{\ket{i} \}_{i=0}^{d-1}$ is orthonormal basis of $\mathbb{C}^{d}$.

(ii) The diagonalization of $W$  shows that its eigenvectors are $\{ \ket{ii}, \ket{\psi_{ij}^{\pm}} = (\ket{ij}\pm \ket{ji})/\sqrt{2}\}$, with corresponding eigenvalues $\{H_{ii}, \pm H_{ij}\}$, i.e.,
\begin{equation}
\label{wtotal}
\begin{split}
	W &= (H^{ext})^{T_{B}}= \sum_{i=0}^{d-1} H_{ii} \ket{ii}\bra{ii} 
	+\\
	&+\sum_{i < j }^{d-1} H_{ij} \ket{\psi_{ij}^{+}}\bra{\psi_{ij}^{+}} - \sum_{i < j }^{d-1} H_{ij} \ket{\psi_{ij}^{-}}\bra{\psi_{ij}^{-}},
\end{split}
\end{equation}
where $\ket{\psi_{ij}^{+}}=\ket{D_{ij}}$ and  $\ket{ii}=\ket{D_{ii}}$.
Notice that the $d(d-1)/2$ eigenvectors corresponding to the projectors $\ket{\psi_{ij}^{-}}\bra{\psi_{ij}^{-}}$, are orthogonal to the symmetric subspace and, therefore, can be discarded by projecting on the symmetric subspace ($\mathcal{S}$). 
\begin{equation}
\label{wsym}
	W _{S}= \Pi_{S}W\Pi_{S}=\sum_{i=0}^{d-1} H_{ii} \ket{D_{ii}}\bra{D_{ii}} 
	+ \sum_{i < j}^{d-1} H_{ij} \ket{D_{ij}}\bra{D_{ij}}. 
\end{equation}
Finally, since copositivity requires that $H_{ii}\geq 0$ $\forall i$,  $W_{S}$ is an EW iff at least one of the remaining eigenvalues is negative, i.e., if $H$ has at least one negative element $H_{mn}=H_{nm}<0$ for some $m\neq n$. It is now trivial to see that  $W_{S}$ indeed detects, at least, the entangled state $\ket{\psi_{mn}^{+}}$ since $\mbox{Tr} (W_{S}\ket{\psi_{mn}^{+}}\bra{\psi_{mn}^{+}} ) = H_{mn} <0$. To conclude, if $W_{S}$ is an EW in the symmetric subspace, so it is $W$ given by \Cref{wtotal}.
\end{proof}

\begin{theorem}
	\label{the2}
	{If $H=H_{\mathcal{N}} + H_{\mathcal{PSD}}$ (i.e., H is non-exceptional) with at least one negative element, then  $W=\Pi_{S} \left( H^{ext}_{\mathcal{N}} \right) ^{T_{B}}\Pi_{S} + \Pi_{S}\left( H^{ext}_{\mathcal{PSD}} \right) ^{T_{B}}\Pi_{S}=\Pi_{S} \left( H^{ext}_{\mathcal{N}} \right) ^{T_{B}}\Pi_{S} +\left( H^{ext}_{\mathcal{PSD}} \right) ^{T_{B}}$ is a decomposable EW. The converse is also true, that is, if $W_S=\Pi_S(H^{ext})^{T_{B}}\Pi_{S}=P+Q^{T_{B}}$ with $P,Q\succeq 0$, then $H = H_{\mathcal{N}} + H_{\mathcal{PSD}}$}.
	\end{theorem}
Notice that it is always possible to decompose a copositive matrix in such a way that the semidefinite positive part $H^{ext}_{\mathcal{PSD}}$ is symmetric.
\begin{proof}
$\Rightarrow$. Assume that $H = 
H_{\mathcal{N}}+ H_{\mathcal{PSD}}$, with $H_{\mathcal{N}} \in \mathcal{N}, ~H_{\mathcal{PSD}} \in \mathcal{PSD}$ and $H_{mn}=H_{nm}<0$. Then
$(H_{\mathcal{PSD}})_{mn} =(H_{\mathcal{PSD}})_{nm}<0$
by construction and $(H^{ext}_{\mathcal{PSD}})^{T_{B}} \nsucceq 0$, but $\Pi_{S} \left( H^{ext}_{\mathcal{N}} \right) ^{T_{B}}\Pi_{S}\succeq 0$. 
The operator $W=\Pi_{S} \left( H^{ext}_{\mathcal{N}} \right) ^{T_{B}}\Pi_{S} + \left( H^{ext}_{\mathcal{PSD}} \right) ^{T_{B}}$ is a decomposable EW: (i) $\braket{e,e|W| e,e} \ge 0$ because $H$ is copositive, (ii) $W$ has at least one negative eigenvalue associated to the negative element $(H_{\mathcal{PSD}})_{mn}$, (iii) $W$ is decomposable since it can be written as $W= P + Q^{T_{B}}$ with $P =\Pi_{S} \left( H^{ext}_{\mathcal{N}} \right) ^{T_{B}}\Pi_{S} \succeq 0 $ and $Q = \Pi_{S} H^{ext}_{\mathcal{PSD}}   \Pi_{S} = H^{ext}_{\mathcal{PSD}}  \succeq 0$.

$\Leftarrow$ Since $W=\Pi_{S} \left( H^{ext}\right)^{T_{B}}\Pi_{S}$ is an EW, there exists at least one negative element of $H$, $H_{mn}=H_{nm}<0$.  We construct 
$Q=H_{mm}|mm\rangle\langle mm| +H_{nn}|nn\rangle\langle nn|+H_{mn}(|mn\rangle\langle mn| +|nm\rangle\langle nm|)$. By construction $Q$ is symmetric and is $Q\notin\mathcal{N}$, $Q^{T_B}\nsucceq 0$ and $Q\succeq 0$, where the last inequality holds because copositivity of $H$ implies  $\sqrt{H_{mm} H_{nn}} \ge -H_{mn}$. Hence, we can identify $Q=H^{ext}_{\mathcal{PSD}}$. Now, define $P=\Pi_S (\sum_{i,j\neq m,n} H_{ij}|ij\rangle\langle ji|)\Pi_S$, which can be expressed as $P=\sum_{i,j\neq m,n} H_{ii}|ii\rangle\langle ii|+\sum_{i,j\neq m,n}
(H_{ij}/2) (|ij\rangle\langle ji|+ |ij\rangle\langle ij|+|ji\rangle\langle ji|+|ji\rangle\langle ij|)$. Clearly  $P\succeq 0, P\in\mathcal{N}$, so that 
$P= \Pi_S (H^{ext}_{\mathcal{N}})^{T_B}\Pi_S$ and
$H = H_{\mathcal{N}} + H_{\mathcal{PSD}}$.
\end{proof}

Let us illustrate \Cref{the2} by considering  
the following copositive matrix in $d=3$ 
\begin{equation}
	H = 
	\begin{pmatrix*}[r]
	1 & 1 & 1 \\
	1 & 1 & -1 \\
	1 & -1 & 1 
	\end{pmatrix*}~.
\end{equation}
A possible  decomposition of $H = H_{\mathcal{PSD}}  + H_{\mathcal{N}}$, with $H_{\mathcal{PSD}} \in \mathcal{PSD}_{3}$ and $ H_{\mathcal{N}} \in\mathcal{N}_{3}$, is given by:
\begin{equation}
H_{\mathcal {PSD}} = 
\begin{pmatrix*}[r]
0 & 0 & 0 \\
0 & 1 & -1 \\
0 & -1 & 1 
\end{pmatrix*}, \quad
H_{\mathcal{N}} = 
\begin{pmatrix*}[r]
1 & 1 & 1 \\
1 & 0 & 0 \\
1 & 0 & 0 
\end{pmatrix*}~.
\end{equation}
\noindent  The associated EW $W= P + Q^{T_{B}}$, with $P = \Pi_{S}(H_{\mathcal{N}}^{ext})^{T_{B}}\Pi_{S}$, and $Q =H_{\mathcal{PSD}}^{ext}$ reads

\begin{equation*}
P = \frac{1}{2}
\begin{pmatrix}
\begin{matrix}
2 & 0 & 0 \\
0 & 1 & 0 \\
0 & 0 & 1 
\end{matrix}
& \rvline & 
\begin{matrix}
0 & 0 & 0 \\
1 & 0 & 0 \\
0 & 0 & 0 
\end{matrix} 
& \rvline & 
\begin{matrix}
0 & 0 & 0 \\
0 & 0 & 0 \\
1 & 0 & 0  
\end{matrix} \\
\hline
\begin{matrix}
0 & 1 & 0 \\
0 & 0 & 0 \\
0 & 0 & 0 
\end{matrix} 
& \rvline &
\begin{matrix}
1 & 0 & 0 \\
0 & 0 & 0 \\
0 & 0 & 0 
\end{matrix}
& \rvline & 
\begin{matrix}
0 & 0 & 0 \\
0 & 0 & 0 \\
0 & 0 & 0  
\end{matrix} \\
\hline
\begin{matrix}
0 & 0 & 1 \\
0 & 0 & 0 \\
0 & 0 & 0 
\end{matrix} 
& \rvline &
\begin{matrix}
0 & 0 & 0 \\
0 & 0 & 0 \\
0 & 0 & 0 
\end{matrix}
& \rvline & 
\begin{matrix}
1 & 0 & 0 \\
0 & 0 & 0 \\
0 & 0 & 0  
\end{matrix} 
\end{pmatrix}~,\quad Q =
\begin{pmatrix}
\begin{matrix*}[c]
0 & 0 & 0 \\
0 & 0 & 0 \\
0 & 0 & 0 
\end{matrix*}
& \rvline & 
\begin{matrix*}[c]
0 & 0 & 0 \\
0 & 0 & 0 \\
0 & 0 & 0 
\end{matrix*} 
& \rvline & 
\begin{matrix*}[c]
0 & 0 & 0 \\
0 & 0 & 0 \\
0 & 0 & 0  
\end{matrix*} \\
\hline
\begin{matrix}
0 & 0 & 0 \\
0 & 0 & 0 \\
0 & 0 & 0 
\end{matrix} 
& \rvline &
\begin{matrix}
0 & 0 & 0 \\
0 & 1 & 0 \\
0 & 0 & 0 
\end{matrix}
& \rvline & 
\begin{matrix}
~0 & 0 & 0 \\
~0 & 0 & \matminus1 \\
~0 & 0 & 0  
\end{matrix} \\
\hline
\begin{matrix}
0 & 0 & 0 \\
0 & 0 & 0 \\
0 & 0 & 0 
\end{matrix} 
& \rvline &
\begin{matrix}
0 & 0 & 0 \\
0 & 0 & 0 \\
0 & \matminus1 & 0 
\end{matrix}
& \rvline & 
\begin{matrix}
0 & 0 & 0 \\
0 & 0 & 0 \\
0 & 0 & 1  
\end{matrix} 
\end{pmatrix}~,
\end{equation*}
\noindent but the resulting EW, $W'=P'+Q'^{T_{B}}$, detects exactly the same states in the symmetric subspace.

The link between non-exceptional copositive matrices and decomposable EWs extends also to exceptional copositive matrices and non-decomposable EWs in the symmetric subspace.
\begin{theorem}
\label{nondeco}
 Associated to each exceptional copositive matrix $H$ (i.e., $H \in \mathcal {COP} \setminus (\mathcal{PSD}+ \mathcal{N})$) with at least one negative entry, there is a non-decomposable EW, $W=(H^{ext})^{T_{B}}$, able to detect symmetric PPTES.
 \end{theorem}

\begin{proof}
For any $H \in \mathcal {COP} \setminus (\mathcal{PSD}+ \mathcal{N})$, $H$ always admits a decomposition of the form  $H= H_{\mathcal{N}}+ H_{\star}$, where $H_{\mathcal{N}}$ is a non-negative symmetric matrix and $H_{\star}$ has at least one negative eigenvalue but is not positive semidefinite. The associated EW $W=P+Q^{T_{B}}$ 
with $P=\Pi_{S}( H^{ext}_{\mathcal{N}})^{T_{B}}\Pi_{S}$ and $Q= H^{ext}_{\star}$, is a non-decomposable EW since $P\succeq 0$ but $Q \nsucceq 0$. The operator $W=(H^{ext})^{T_{B}}$ is also a non-decomposable EW. 
\end{proof}

\begin{corollary}[From \cite{tura2018separability}]
	\label{d4}
	Since for $d<5$ every copositive matrix is not exceptional (i.e., $H = H_{\mathcal{PSD}}+ H_{\mathcal{N}}$), all EWs of DS states in $d=3$ and $d=4$ are decomposable.
\end{corollary}

The above corollary rephrases the fact that PPT criterion is necessary and sufficient to assess separability for bipartite DS states $\rho_{DS} \in \mathcal{B}(\mathbb C^d\otimes \mathbb C^d)$ for $d<5$.

\section {Symmetric PPTES}
\label{sec:symBoundEntanglement}
\noindent 
Let us briefly summarize what we have seen so far. The fact that each DS state, $\rho_{DS}\in \mathcal{B}(\mathbb C^d\otimes \mathbb C^d)$, is associated to a matrix $M_d(\rho_{DS})$ (see  Eq.\eqref{mr}), allows to reformulate the problem of entanglement characterization as the equivalent problem of checking the membership of $M_d(\rho_{DS})$ to the cone of completely positive matrices $\mathcal{CP}_d$. Equivalently, according to the dual formulation, any entangled state $\rho_{DS}$ is detected by an EW $W$ which can be constructed from a copositive matrix $H$. PPT entangled diagonal symmetric states (PPTEDS) can only be detected by non-decomposable EWs, which correspond to exceptional copositive matrices. Since for $d<5$, all copositive matrices $H$ are of the form $H = H_{\mathcal{PSD}}+ H_{\mathcal{N}}$, all EWs defined as $W=(H^{ext})^{T_B}$ are necessarily of the form $W=P+Q^{T_B}$, with $P,Q\succeq 0$, meaning that for $d<5$ there are not PPTEDS.

 However, for $d>5$, this is not the case, since there exist exceptional copositive matrices, i.e., $H \notin \mathcal{PSD}_{d}+ \mathcal{N}_{d}$. Thus, detecting entanglement of $\rho_{DS}$ in $d\geq 5$, is equivalent to checking membership of the corresponding copositive matrix $H \in \mathcal {COP}_{d} \setminus (\mathcal{PSD}_{d}+ \mathcal{N}_{d})$, which is, in general, a co-NP-hard problem \cite{murty1985some}. 
 
What can we say about symmetric PPTES $\rho_S$ that are not DS? In this section, we tackle the problem of entanglement detection for generic states  $\rho_{S} \in \mathcal{B}(\mathbb{C}^{d}\otimes \mathbb{C}^{d})$ in arbitrary dimension $d$. 
Following the argument given above, we split our analysis in two different scenarios, namely when $d\ge 5$ and $d<5$. Remarkably, even outside of the DS paradigm, we find that copositive matrices lie at the core of non-decomposable EWs for symmetric PPTES in arbitrary dimensions. 

\subsection {Symmetric PPTES in $d\ge 5$}

\noindent 
The fact that for $d\geq 5$ there exist exceptional copositive matrices with at least one negative entry which lead to non-decomposable EWs in $\mathcal{S}$, implies that (i) such EW can detect a PPTEDS, and (ii) the same EW is able to detect other symmetric, but not DS, PPTES "around" it.

\begin{theorem}
	\label{th}
	Let $\rho_{DS}$ be a PPTEDS. Then any symmetric state $\rho_{S}  = \rho_{DS} + \sigma_{CS}$, such that 
	$\rho_{S}^{T_B} \geq 0$, is PPT entangled.
\end{theorem}
\begin{proof} 
Since $\rho_{DS}$ is a PPTEDS state there exists an exceptional copositive matrix $H$ and an associated non
decomposable EW $W$ such that $\mbox{Tr}(W \rho_{DS})<0$. It follows that $\mbox{Tr} (W\rho_{S})= \mbox{Tr}(W(\rho_{DS}+\sigma_{CS}))=\mbox{Tr}(W \rho_{DS})=\mbox{Tr}(H M_{d}(\rho_{DS}) ) <0$, so that $\rho_{S}$ is PPT entangled.
\end{proof}

\noindent The paradigmatic example of an exceptional copositive matrix in $d=5$, is the so-called Horn matrix \cite{hall1963copositive} which is the matrix associated to the quadratic form $\vec{x}^{T} H \vec{x} = (x_{1}+x_{2}+x_{3}+x_{4}+x_{5})^2 - 4x_{1}x_{2}- 4x_{2}x_{3}- 4x_{3}x_{4} - 4x_{4}x_{5}- 4x_{5}x_{1}$~. 
Exceptional copositive matrices of the Horn type,  $H_{\mathcal{H}}$, can be generated for any odd $d\ge 5$ \cite{johnson2008constructing}, and are of the form
\begin{equation}
\label{eq:hornodd}
H_{\mathcal{H}} = 
\begin{pmatrix*}[r]
1  & -1 & 1 & 1 & \cdots & \cdots & 1 & 1 & 1 & -1 \\
-1 & 1  & -1 & 1 & \ddots & & \ddots & 1 & 1 & 1 \\
1 & -1 & 1 & -1 & 1 &  &  & \ddots &  1 & 1 \\
1 & 1 & -1 & 1 & -1 & \ddots & & & \ddots &  1 \\
\vdots & \ddots & 1 & -1 & 1 & \ddots & & & & \vdots \\
\vdots & & & \ddots & \ddots & \ddots & \ddots & \ddots & & \vdots\\
1 & \ddots & & & & \ddots & 1 & -1 & 1 & 1 \\
1 & 1 & \ddots & & & \ddots & -1 & 1 & -1 & 1 \\
1 & 1 & 1 & \ddots & & & 1 & -1 & 1 & -1 \\
-1 & 1 & 1 & 1 & \cdots & \cdots & 1 & 1 & -1 & 1
\end{pmatrix*}~.
\end{equation}
\noindent Since the $H_{\mathcal{H}}$ is exceptional and has negative entries, it leads to a non-decomposable EW $W= (H^{ext})^{T_{B}}$, that can be used to detect PPTEDS in any odd dimension $d \ge 5$. Moreover, due to Th.(\ref{th}), by adding suitable coherences to such states, the same EW can be used to certify PPT-entanglement also in whole families of symmetric states. Below we provide one of these families.

\begin{corollary}
\label{ijppt}
Given a PPTEDS state, $\rho_{DS}$, any symmetric state of the form $\rho_{S}= \rho_{DS} + \sigma_{CS}$, with  $\sigma_{CS}=\sum_{i<j} (\alpha_{ij}\ket{D_{ii} }\bra{D_{jj}}+\emph{h.c.})$ and $|\alpha_{ij}| \leq \frac{p_{ij}}{2}$ is PPT-entangled.
\end{corollary}
\begin{proof}

\noindent The state, $\rho_{S}$, and its partial transpose, $\rho_{S}^{T_{B}}$, can be cast  as
\begin{align}
    \label{rhoco}
&\rho_{S\;} = \tilde{M}_{d}(\rho_{S}) \bigoplus_{i<j} \frac{p_{ij}}{2} \begin{pmatrix*}
1 & 1 \\
1 & 1  
\end{pmatrix*}~,\\
&\rho_{S}^{T_{B}} = M_{d}(\rho_{DS}) \bigoplus_{0\le i < j < d} 
\begin{pmatrix*}
p_{ij}/2 & \alpha_{ij}\\
\alpha_{ij}^* & p_{ij}/2
\end{pmatrix*}~,
\end{align}
with 
\begin{equation*}
\tilde{M}_{d}(\rho_{S}) =
\begin{pmatrix*}
p_{00} & \alpha_{01} &  \cdots & \alpha_{0,d-1} \\
\alpha_{01}^* & p_{11} &  \cdots & \alpha_{1,d-1}\\
\vdots & \vdots & \vdots & \ddots &  \\
\alpha_{0,d-1}^* & \alpha_{1,d-1}^* & \cdots &p_{d-1,d-1}  
\end{pmatrix*}~,
\end{equation*}
\begin{equation*}
M_{d}(\rho_{DS}) =
\begin{pmatrix*}
p_{00} & p_{01}/2 &  \cdots & p_{0,d-1}/2 \\
p_{01}/2 & p_{11} &  \cdots & p_{1,d-1}/2\\
\vdots & \vdots & \vdots & \ddots &  \\
p_{0,d-1}/2 & p_{1,d-1}/2 & \cdots & p_{d-1,d-1}  
\end{pmatrix*}~.
\end{equation*}

\noindent Positive semidefiniteness of ${\rho_{S}}^{T_{B}}$ implies
$|\alpha_{ij}|\le \frac{p_{ij}}{2}$, so that the state $\rho_{S}$, generated from a PPTEDS state, remains PPT-entangled -- since it is detected by the same non-decomposable EW -- as long as the coherences respect the condition $|\alpha_{ij}| \leq \frac{p_{ij}}{2}$.
\end{proof}

A further connection between copositive matrices and EWs appears when considering {\it {extreme}} copositive matrices. For instance, let us consider the (generalized) Horn matrix $H_{\mathcal{H}}$ of Eq.\eqref{eq:hornodd}, and the so-called Hoffmann-Pereira matrix $H_{\mathcal{HP}}$ \cite{johnson2008constructing,hoffman1973copositive}, which, besides of being exceptional, is also extreme. For $d=7$, such copositive matrices take the form

\begin{equation}
H_{\mathcal{H}} =
\begin{pmatrix*}[r]
1 	&-1 & 1  &  1  & 1  & 1 & -1\\
-1 	& 1 &-1 &  1  & 1  &1  & 1 \\
1   &-1 & 1  & -1 & 1  &1  & 1\\
1 	& 1  &-1 &  1  & -1 &1  & 1 \\
1 	& 1  & 1  & -1 & 1  &-1  & 1\\
1 	& 1  & 1  & 1 & -1  &1  & -1\\
-1 	& 1  & 1  & 1 & 1  &-1  & 1
\end{pmatrix*}~,
\end{equation}
\begin{equation}
H_{\mathcal{HP}}=
\begin{pmatrix*}[r]
1 	&-1 & 1  &  0  & 0  & 1 & -1\\
-1 	& 1 &-1 &  1  & 0  &0  & 1 \\
1   &-1 & 1  & -1 & 1  &0  & 0\\
0 	& 1  &-1 &  1  & -1 &1  & 0 \\
0 	& 0  & 1  & -1 & 1  &-1  & 1\\
1 	& 0  & 0  & 1 & -1  &1  & -1\\
-1 	& 1  & 0  & 0 & 1  &-1  & 1
\end{pmatrix*}~.
\end{equation}
Let us inspect the action of both matrices, $H_\mathcal{H}$ and 
$H_\mathcal{HP}$, on a DS state $\rho_{DS}\in\mathcal{B}( \mathbb C^{7}\otimes \mathbb{C}^{7})$, described by its associated $M_7(\rho_{DS})$ (see Eq.\eqref{mr}):
\begin{equation}
M_{d}(\rho_{DS})=
\label{PPTEDS7}
\begin{pmatrix*}[c]
1 	& 1  & 0  & 0 & 0  &0   & 1/8\\
1 	& 2  &  1 &  0  & 0  &0   & 0\\
0  & 1  &  2  & 1  &0   &0   & 1/4\\
0 	& 0  & 1 &  2  & 1  &0   & 0\\
0 	& 0  & 0  & 1  & 2  &1   & 0\\
0 	& 0  & 0  & 0  & 1  &2   & 1\\
1/8	 & 0  & 1/4   & 0  &0   & 1  & 1
\end{pmatrix*}~.
\end{equation}
It can be easily checked that 
$\mbox{Tr}( H_\mathcal{HP}M_{d}(\rho_{DS}))= \mbox{Tr}(\left(H^{ext}_{\mathcal{HP}}\right)^{T_{B}} \rho_{DS})= -\frac{1}{4}$. Since both $H_{\mathcal{H}}$ and $H_\mathcal{HP}$ are exceptional copositive matrices,  $W_{\mathcal{HP}}=\left(H^{ext}_{\mathcal{HP}}\right)^{T_{B}}$  and $W_{\mathcal{H}}=\left(H^{ext}_{\mathcal{H}}\right)^{T_{B}}$ are non-decomposable EWs, so that $\rho_{DS}$ is a PPT entangled.

 In contrast,  $\mbox{Tr}( H_\mathcal{H}M_{d}(\rho_{DS}))=\mbox{Tr}( W_{\mathcal{H}} \rho_{S})=0$, indicating that $H_{\mathcal{H}}$ fails to detect this state. Moreover, as stated by Th.\eqref{th}, $W_\mathcal{HP}= (H^{ext}_{\mathcal{HP}} )^{T_{B}}$, detects, as well, many other states around the state given by Eq.\eqref{PPTEDS7}. 
Given the relationship between non-decomposable EWs and exceptional matrices, we conjecture that extremality in the copositive cone correspond to optimality in the set of EWs. In other words, copositive matrices that are both extreme and exceptional lead to optimal non-decomposable EWs in the sense of \cite{lewenstein2000optimization}. We complete this subsection with the following theorem regarding extreme copositive matrices (see Fig.\ref{fig:cop})

\begin{theorem}
\label{extreme}
Let $H$ be an extreme copositive matrix with at least one negative eigenvalue, and at least one negative element $H_{ij}<0$. Then $H$ must be exceptional.
\end{theorem}
\begin{proof}
$H$ cannot belong to neither $ H_{PSD}$ nor to $H_N$ and, while it is extremal, it cannot be a combination of their elements as well. A more detailed proof can be found in \Cref{sec:appendixB}.
\end{proof}

\subsection {Symmetric PPTES in $d<5$}

 In what remains, we are interested in symmetric PPTES of the form $\rho_{S}=\rho_{DS}+ \sigma_{CS}$ where $\rho_{DS}$ is separable, so that $\mbox{Tr}( H M_{d}(\rho_{DS}))\geq 0$ for all copositive matrices $H$. Moreover, since for $d < 5$, every copositive matrix is non-exceptional, i.e., $H=H_{\mathcal{N}} + H_{\mathcal{PSD}}$, the corresponding witness $W= (H^{ext})^{T_{B}}$ will always be decomposable.  For this reason, coherences are needed to create PPTES in low dimensional systems. Here we show that such states symmetric PPTES can nevertheless be detected by EWs which are of the form $W_{S}= W + W_{CS}$, that is by adding to the decomposable EW, $W$, a convenient off-diagonal, symmetric contribution $W_{CS}$ which reads the coherences of $\rho_S$.
 
For the sake of simplicity, we hereby consider symmetric states of the form
\begin{equation}
\label{coho}
\rho_{S}= \rho_{DS} + \sigma_{CS} = \sum_{ij} p_{ij}\ket{D_{ij}}\bra{D_{ij}} + \sum_{i\neq j\neq k} (\alpha_{ijk} \ket{D_{ii}}\bra{D_{jk}})+ \mbox{h.c.}~,
 \end{equation}
 with $p_{ij}\geq 0$ $\forall i,j \;$, $\sum p_{ij}=1$ and $\alpha_{ijk}\in \mathbb C$.
 
 Indeed, in this case, both $\rho_{S}$ and $\rho_{S}^{T_{B}}$ can be cast as a direct sum of matrices, which highly simplifies our analysis. For instance, for $d=3$, $\rho_{S}$ and $\rho_{S}^{T_{B}}$ are of the form
\begin{equation}
\label{rankrho}
\rho_{S} =
\begin{pmatrix*}[c]
\frac{p_{02}}{2} & \alpha & \frac{p_{02}}{2} \\
\alpha^{*} & p_{11} & \alpha \\
\frac{p_{02}}{2} & \alpha^{*} & \frac{p_{02}}{2}
\end{pmatrix*}
\oplus 
\begin{pmatrix*}[c]
p_{00} & \beta & \beta\\
\beta^{*} &  \frac{p_{12}}{2} & \frac{p_{12}}{2} \\
\beta^{*} &  \frac{p_{12}}{2}  & \frac{p_{12}}{2} 
\end{pmatrix*}
\oplus
\begin{pmatrix*}[c]
\frac{p_{01}}{2} & \frac{p_{01}}{2} & \gamma \\
\frac{p_{01}}{2} &  \frac{p_{01}}{2} & \gamma \\
\gamma^{*} &  \gamma^{*}  & p_{22} 
\end{pmatrix*}~,
\end{equation}

\begin{equation}
\label{rankrhotb}
	\rho_{S}^{T_{B}} =
	M_{d}(\rho_{DS}) 
	\oplus 
	\begin{pmatrix*}[c]
	\frac{p_{01}}{2} & \alpha &\beta \\
 \alpha^{*} &  \frac{p_{12}}{2} &  \gamma \\
	\beta^{*} &   \gamma^{*}  & \frac{p_{02}}{2} 
	\end{pmatrix*} 
	\oplus
	\begin{pmatrix*}[c]
	\frac{p_{02}}{2} & \beta & \gamma \\
	\beta^{*}  &  \frac{p_{01}}{2} & \alpha \\
	\gamma^{*} &  \alpha^{*} & \frac{p_{12}}{2} 
	\end{pmatrix*}~,
\end{equation}
\noindent where we have defined, for easiness of reading, $\alpha \equiv \alpha_{120} = \alpha_{102}$, $\beta \equiv \alpha_{012}= \alpha_{021}$ and $\gamma \equiv \alpha_{201} = \alpha_{210}$. Such structure, which corresponds to a particular direct sum decomposition of the total Hilbert space, bears similitude with the so-called circulant states \cite{PhysRevA.76.032308}.

In order to investigate the existence of PPTES we focus on edge states, since their low-dimensional ranks allow for a simpler analysis. 
By using a notation common in the literature, we say that an edge state $\rho_{S}$ is of type $(p,q)$ if $p=r(\rho_{S})$ and $q=r(\rho_{S}^{T_{B}})$ are the ranks of $\rho_{S}$ and $\rho_{S}^{T_{B}}$, respectively.
 While symmetric states in $d=3$ are, generically, of type $(6,9)$, PPT-entangled edge states must have lower ranks. When dealing with states of the form of Eq.(\ref{coho}), we have found numerically that at least two coherences must be considered.
 For instance, we can set $\gamma = 0$ and choose $\alpha$ and $\beta$ in such a way to lower the value of $(p,q)$. 
Indeed, a direct inspection of Eqs.(\ref{rankrho})-(\ref{rankrhotb}), shows that, by setting $|\alpha|^2= p_{11}p_{02}/2$ and $|\beta|^2= p_{02}(p_{01}p_{12}-2p_{02}p_{11})/4 p_{12}$, it is possible to attain a state of type $(5,7)$. Now, starting from a copositve matrix $H$, we can construct a non-decomposable EW of the form
\begin{equation}
\label{wits}
W_{S}= \Pi_{S}(H^{ext})^{T_{B}}\Pi_{S} + \sum_{i\neq j \neq k} W^{i}_{jk}\ket{D_{ii}} \bra{D_{jk}}+ \mbox{h.c.}~,
\end{equation}
where the coefficients $W^{i}_{jk}$ can be chosen to be real. \\

Let us illustrate the above results by providing some explicit examples. We first consider the symmetric edge PPTES provided in \cite{toth2010separability}. The state is of the form in Eq.(\ref{coho}) for $d=3$ (\textit{i.e.,} of the form in Eq.(\ref{rankrho})), obtained from the DS state $\rho_{DS}$ with parameters $p_{00}=0.22, p_{11}=0.234/3, p_{22}=0.183, p_{01}=0.176, p_{02}=0.167/3, p_{12}=0.254$ and coherences $\alpha = 0.167/3, \beta = -0.059/\sqrt{2}, \gamma = 0 $.
Its entanglement was previously certified in \cite{toth2010separability} by means of the PPT symmetric extension proposed in \cite{doherty2004complete}.
With the aid of the same method \cite{doherty2002distinguishing}, we can derive as well the corresponding non-decomposable EW $W$ via semidefinite programming. The associated copositive matrix, $H$, is easily obtained from the properly symmetrized EW, i.e.,
$W_S=\Pi_S W \Pi_S$, and has the form
\begin{equation}
\label{with1}
H \approx
\left(
\begin{matrix}
0.003 &  10.39 & 100.57 \\
10.39 & 59.31 & -21.02 \\
 100.57 &-21.02 & 14.22
\end{matrix}
\right)~,
\end{equation}
while the coefficients $W^{i}_{jk}$ are given by $W^{1}_{02} =23.20$ and $W^{0}_{12} = -37.40$. 
If we restrict to the DS part of the state $\rho_S$, it is trivial to check that $\mbox{Tr}(H M_{d}(\rho_{DS}))\geq 0$. This is by no means a surprise, since for DS states, in $d<5$, the PPT condition implies separability.
For this reason, the coherences provided by the term $\sigma_{CS}$ are necessary to induce the PPT entanglement. 
Remarkably, one can vary the value of the coherences $\alpha$ and $\beta$ to obtain other symmetric PPTES as certified by the EW, i.e., $\mbox{Tr}(W_{S}\rho_{S})<0$. In fact, the EW $W_{S}$ can be used to derive families of PPT entangled states obtained by adding to the state $\rho_{S}$ any coherent contribution $\sigma_{CS}$ of the form of Eq.(\ref{rankrho}) that preserves the positivity of both the state and its partial transpose. Indeed, also in the case $\gamma \neq 0$, as long as the conditions $\rho_{S} \ge 0$ and $\rho_{S}^{T_{B}} \ge 0$ hold, the same non-decomposable EW $W_{S}$, is able to detect, for suitable values of its entries $W^{i}_{jk}$, a whole family of PPTES of the form of Eq.(\ref{rankrho}), as depicted in Fig.\ref{fig:wgamma}.

\begin{figure}[H]
	\centering
	\includegraphics[width=10cm]{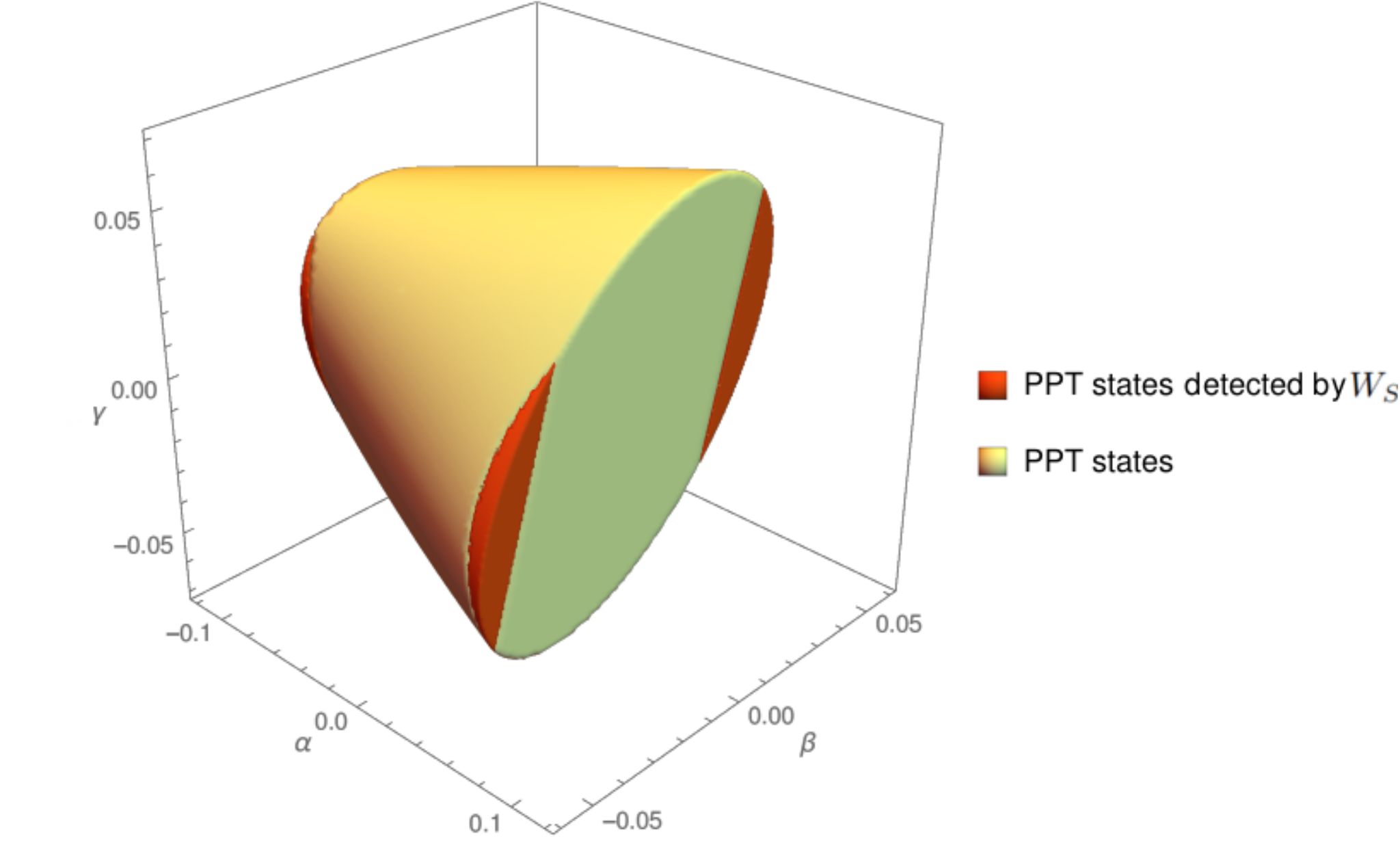}
	\caption{
	The PPT entangled states detected by $W_{S}$ of Eq.(\ref{wits}) with coefficients $|W^{1}_{02}| =23.20, |W^{0}_{12}| = -37.40$ (dark orange) as compared to the whole family of PPT states $\rho_{S}$ of Eq.(\ref{rankrho}) with $p_{00}=0.22, p_{11}=0.234/3, p_{22}=0.183, p_{01}=0.176, p_{02}=0.167/3, p_{12}=0.254$ (light orange). }
	\label{fig:wgamma}
\end{figure}

\begin{figure}[H]
	\centering
	\includegraphics[width=10cm]{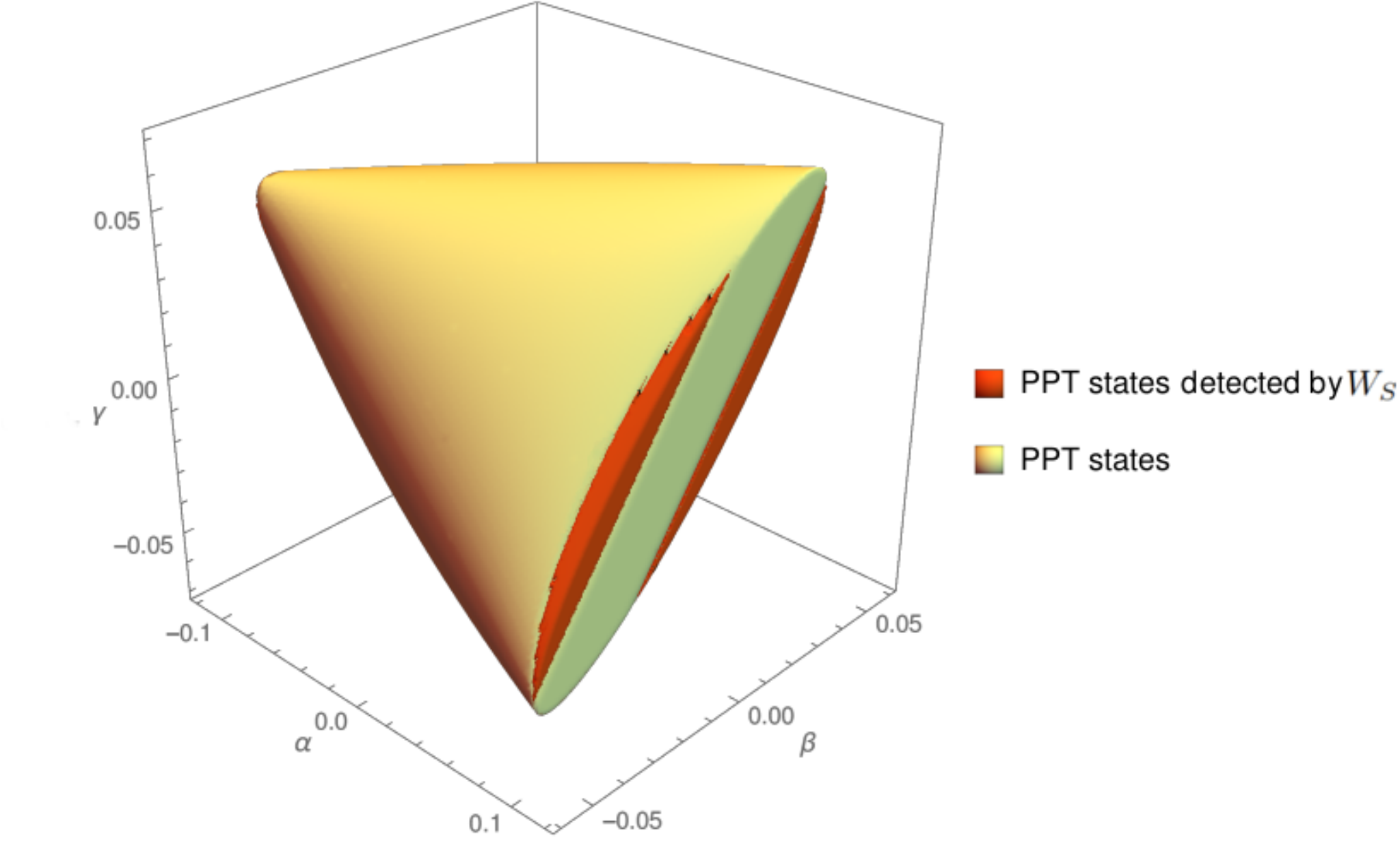}
	\caption{The PPT entangled states detected by $W_{S}$ of Eq.(\ref{wits}) with coefficients $|W^{1}_{02}| =\frac{4595}{191}, |W^{0}_{12}| = \frac{6114}{113}$ (dark orange) as compared to the whole family of PPT states $\rho_{S}$ of Eq.(\ref{rankrho}) with  $p_{00}= p_{11}= p_{12} = \frac{1848}{7625}$, $p_{22}= \frac{464}{7625}$, $p_{01}=\frac{231}{1525}$, $p_{02}=\frac{462}{7625}$ (light orange). }
	\label{fig:wgamma2}
\end{figure}
In Fig.\ref{fig:wgamma2}, we display a new example of a symmetric PPTES $\rho_{S}$ of the form of Eq.(\ref{rankrho}), found by semidefinite programming. Also in this case, we have found a non-decomposable EW $W_{S}$ of the form of Eq.(\ref{wits}), with coefficients $W^{1}_{02} =\frac{4595}{191}$ and $W^{0}_{12} = -\frac{6114}{113}$ and whose associated copositive matrix is given by

\begin{equation}
\label{with2}
H =
\left(
\begin{matrix}
 1/172 & 1009/151 & 11025/68 \\
 1009/151 & 1803/22 & -5829/65 \\
 11025/68 & -5829/65 & 1224/7
\end{matrix}
\right)~.
\end{equation}

Similarly, we can use the same procedure to derive families of PPT entangled symmetric states for $d=4$. In this case, we have found, numerically, that at least three different coherences of the form of Eq.(\ref{coho}) are needed in order to get a low-dimensional PPT entangled edge state. To the best of our knowledge, there are no explicit examples of symmetric PPT entangled states in $d=4$.
We provide an example in \Cref{sec:appendixA} together with the non-decomposable EW that detects it and its associated copositive matrix.

\section{Conclusions}
\label{sec:Conclusions}
In this work we have studied the connection between EWs and copositive matrices, showing how this class of matrices can be effectively used to detect entanglement in bipartite symmetric states.  First, we have proved that non-decomposable (decomposable) EWs can be derived from exceptional (non-exceptional) copositive matrices. Second, we have tackled the problem of the entanglement detection for two symmetric qudits. In this context, our analysis has shown that, on the one hand, for dimension $d \ge 5$, it is possible to certify bound entanglement in any family of symmetric states constructed by adding coherences to a PPTEDS state. Indeed, we have provided the explicit expression of a symmetric PPTES in $d=7$, along with the exceptional matrix that detects it. On the other hand, for dimension $d<5$, every copositive matrix is not exceptional, so that a different approach is needed in order to construct a non-decomposable EW. Nevertheless, we have shown that, also in this case, an EW of this type can be derived from a not exceptional copositive matrix by adding suitable coherent-like contributions. Indeed, we conjecture that $\textit{any}$ symmetric PPTES of two qudits can be detected by a non-decomposable EW of the form of Eq.(\ref{wits}). Our conjecture seems to be well grounded as pointed out by the several examples we have found of symmetric PPTES detected with this method, both for $d=3$ and $d=4$ (see Appendix \ref{sec:appendixA} for further details).

\appendix

\section{A symmetric PPTES in $d=4$}
\label{sec:appendixA}

Here we present an explicit example of a PPTES $\rho_S=\rho_{DS}+\sigma_{CS}$ of the form of Eq.(\ref{coho}) for $d=4$. The state is given by: i) a DS state $\rho_{DS}$ with $p_{00}= p_{02}= p_{03}=p_{11}=p_{22} = \frac{172+16 \sqrt{2}}{1817}$, $p_{01}=p_{13}=\frac{32+172 \sqrt{2}}{1817}$, $p_{11}=p_{12}=p_{23}=\frac{86+8\sqrt{2}}{1817}$, $p_{33}=\frac{721-440\sqrt{2}}{1817}$; and ii) a coherence term $\sigma_{CS}$ with $\alpha =p_{00}, \beta=-p_{01}/2$ and $\gamma=p_{01}/4$.

Again, to certify its entanglement we have used the symmetric extension \cite{doherty2002distinguishing}, which provides a non-decomposable EW, $W_S$, via semidefinite programming. Such EW is of the form of Eq.(\ref{wits}), with coefficients $W^{0}_{23}= \frac{6526}{321}, W^{1}_{03} = -\frac{1896}{107}, W^{2}_{13}=-\frac{549}{1238}$ and has an associated copositive matrix
\begin{equation}
\label{with3}
H =
\left(
\begin{matrix}
 21/3590 & 9425/1571 & 4853/464 & 1111/28 \\
 9425/1571 & 1293/88 & 2122/145 & 220/323 \\
 4853/464 & 2122/145 & 6/5951 & 1355/3014 \\
 1111/28 & 220/323  & 1355/3014 &  862/7403
\end{matrix}
\right)~.
\end{equation}
Let us observe that, despite the fact that $H$ of Eq.(\ref{with3}) does not have any negative matrix element, the corresponding $EW$ has nevertheless a negative eigenvalue. This observation makes clearer, once more, the fact that, in $d<5$, differently from the case $d\ge5$, the possibility to detect a PPTES relies exclusively on a convenient choice of the coherences $W^{i}_{jk}$.

\section{An alternative proof of \Cref{extreme}}
\label{sec:appendixB}
\begin{proof}
Since $H$ is an extreme copositive matrix, it only admits the trivial decomposition $H = H_{1}+ H_{2}$ with $H_{1} = a H$ and $H_{2}= (1-a) H$ copositive, for every $a \in [0,1]$. Let $H= H_{\mathcal{PSD}} + H_{\mathcal{N}}$ where $H_{\mathcal{PSD}}$ is a positive semidefinite matrix and $H_{\mathcal{N}}$ is a non-negative matrix. We are left with three possibilities: i) $H_{1} = H_{\mathcal{PSD}}$ and $H_{2} = H_{\mathcal{N}}$, ii) $H_{1} = H_{\mathcal{N}}$ and $H_{2} = H_{\mathcal{PSD}}$, iii) $H = H_{1}+H_{2}= H_{\mathcal{N}}$. i) Since $H$ has at least one negative eigenvalue, the same holds true also for $H_{1}=a H$, so that $H_{1}$ is not positive semidefinite. ii) An analogous consideration on $H_{2}$ leads to the same conclusion also in this case, so that $H$ is exceptional. iii) $H = H_{1}+H_{2}= H_{\mathcal{N}}$. Since $H$ has at least one negative entry $H_{ij}<0$ it cannot be a non-negative matrix.
\end{proof}

\acknowledgments{We acknowledge financial support from ERC-AdG NOQIA; Spanish Agencia Estatal de Investigación (AEI): Severo Ochoa Grant No. CEX2019-000910-S, PID2019-107609GB-100, PID2019-106901GB-I00/10.13039 / 501100011033, MINCIN-EU QuantERA MAQS PCI2019-111828-2, 10.13039/501100011033; Generalitat de Catalunya: AGAUR 2017-SGR-1341,2017-SGR-1127, CERCA Program and  QuantumCAT 001-P-001644, QuantumCAT U16-011424 co-funded by ERDF Operational Program of Catalonia 2014-2020; Fundaci\'o Privada Cellex; Fundaci\'{o} Mir-Puig; EU Horizon 2020 FET-OPEN OPTOLogic: Grant No. 899794; and the National Science Centre Poland-Symfonia Grant No. 2016-20-W-ST4-00314.
J.T. thanks the Alexander von Humboldt foundation for support.}

\bibliographystyle{unsrtnat}
\bibliography{copositive}

\end{document}